\documentclass[12pt]{article}
\usepackage{enumitem}
\usepackage{amsmath}
\usepackage{amsthm}
\usepackage{fancyhdr}
\usepackage{mathrsfs}
\usepackage{natbib}
\usepackage[margin=1in]{geometry}
\usepackage{color}
\usepackage{multirow,array}
\usepackage{amsmath,amsthm,amssymb}
{
\theoremstyle{plain}
\newtheorem{theorem}{Theorem}[section] \newtheorem{proposition}{Proposition}[section] 
\newtheorem{assumption}{Assumption}

\newtheorem{remark}{Remark}

\newtheorem{corollary}{Corollary}
  }
\usepackage[colorlinks = true,
            linkcolor = blue,
            urlcolor  = blue,
            citecolor = blue,
            anchorcolor = blue]{hyperref}
\usepackage{hyperref}
\usepackage{verbatim}
\usepackage{tikz}
\usepackage{mathtools}

\usepackage{appendix}
\usepackage{float}
\usepackage{graphicx}
\usepackage{caption}
\usepackage{subcaption}
\usepackage[utf8]{inputenc}
\usepackage[T1]{fontenc}    
\usepackage{hyperref}       
\usepackage{url}            
\usepackage{booktabs}       
\usepackage{amsfonts}       
\usepackage{nicefrac}       
\usepackage{microtype}      

\usepackage{dsfont}
\usepackage{indentfirst}

\usepackage{adjustbox}

\numberwithin{equation}{section}

\title{Identification via Distributional Shifts without Exclusion Restrictions
}
\author{Xunkang Tian \footnote{European Research University. Email: \href{xunkang.tian@eruni.org}{xunkang.tian@eruni.org}} \\
Nan Zhi \footnote{University of Florida. Email: \href{nan.zhi@ufl.edu}{nan.zhi@ufl.edu}}
}
\date{\today}

\begin{document}

\maketitle

\begin{abstract}
This paper studies identification and inference in a triangular system with an endogenous regressor when exclusion restrictions are unavailable and the dependence between structural disturbances is modeled through an unrestricted control function. In this setting, standard orthogonality conditions do not deliver point identification, as the unknown control function can rationalize a wide range of structural coefficients. We show that identifying information can be extracted from distributional shifts in the first-stage disturbance induced by an auxiliary variable that may directly affect the outcome. Imposing a local restriction on the log density ratio, together with an explicit bound on the sieve approximation error of the control function, we derive moment inequalities that restrict the structural parameter. We develop a practical inference procedure based on test inversion and multiplier bootstrap that accommodates generated regressors, cross-fitted sieve estimation, and locally estimated density-ratio nuisances. The results clarify how identification can be recovered from weak local distributional structure in the absence of classical instruments.
\end{abstract}

\begin{center}
\end{center}

\clearpage

\section{Introduction}\label{sec:intro}

A central problem in applied econometrics is the identification of causal effects when an endogenous regressor is correlated with the structural disturbance and a credible excluded instrument is unavailable. Classical instrumental variables strategies rely on exclusion restrictions, but in many empirical settings such restrictions are difficult to justify. This paper studies what can be learned about the coefficient on an endogenous regressor in a triangular system when the researcher is unwilling to impose a classical exclusion restriction and allows for flexible, potentially nonlinear dependence between structural disturbances.

We consider a standard triangular framework in which endogeneity arises through an unknown control function linking the disturbances in the outcome and first-stage equations. The control function is left nonparametric and unrestricted. Without additional structure, the coefficient of interest is generally not point identified, since the unknown control function can absorb a wide range of values of the structural parameter.

A large literature has proposed alternative identification strategies in the absence of conventional instruments. One approach exploits heteroskedasticity to construct internal instruments \citep{lewbel2012using}. Related work studies identification in triangular systems without exclusion restrictions under structured assumptions on error variances \citep{klein2010estimating}. Our approach differs from these contributions. Rather than imposing global parametric structure to recover point identification directly, we derive informative moment inequalities from weak, local restrictions on how an auxiliary variable shifts the distribution of the first-stage disturbance.

Specifically, we allow an auxiliary variable to affect the distribution of the first-stage shock, but not require it to satisfy an exclusion restriction in the outcome equation. We model this distributional shift through a local restriction on the log density ratio between groups. Economically, this restriction can be interpreted as a local change in dispersion or volatility. Technically, it corresponds to a local exponential-tilting approximation with an explicit sensitivity parameter that bounds misspecification.

We combine this distributional structure with a local sieve approximation of the unknown control function. By bounding the sieve approximation error and exploiting orthogonality conditions implied by the triangular structure, we obtain a set of moment inequalities that restrict the structural parameter. These inequalities become informative when the distributional shift is locally relevant and the control function can be well approximated. We show that the identified set shrinks as the sieve approximation error decreases, and under a local relevance condition, point identification is obtained in the limit.

For inference, we construct sample analogues of the moment inequalities using generated first-stage residuals, cross-fitted local sieve regression, and a locally estimated density ratio. Confidence sets are obtained by inverting a max-type test statistic with multiplier bootstrap critical values and intersecting across a finite grid of tuning parameters. The procedure accommodates generated regressors and nuisance estimation while maintaining valid asymptotic coverage.

The contribution of the paper is twofold. First, it provides a new way to extract identifying content in triangular systems with an unrestricted control function, using distributional shifts rather than exclusion restrictions. Second, it delivers a fully operational inference procedure that transparently incorporates approximation error and local misspecification through sensitivity parameters. The results clarify how much identification can be recovered from weak distributional structure and how this identification strength depends on approximation quality and local relevance.

\section{Triangular System with Flexible Dependence}

\subsection{Model Setup}

We consider the following triangular system:
\begin{align}
Y_1 &= X'\beta_1 + \gamma_1 Y_2 + \varepsilon_1, \label{eq:secondstage}  \\
Y_2 &= X'\beta_2 + \varepsilon_2, \label{eq:firststage}
\end{align}
where $X \in \mathbb{R}^k$ is a vector of exogenous covariates,
$Y_2$ is a potentially endogenous regressor,
and $(\varepsilon_1,\varepsilon_2)$ are structural disturbances.

The key feature of the model is that we allow for a flexible, possibly nonlinear dependence between $\varepsilon_1$ and $\varepsilon_2$. Specifically, we assume
\begin{equation}
\varepsilon_1 = h(\varepsilon_2) + \eta,
\label{eq:control-decomp}
\end{equation}
where $h:\mathbb{R}\to\mathbb{R}$ is an unknown measurable function and $\eta$ is a remainder disturbance.

Substituting \eqref{eq:control-decomp} into \eqref{eq:secondstage}, the structural equation becomes
\begin{equation}
Y_1 = X'\beta_1 + \gamma_1 Y_2 + h(\varepsilon_2) + \eta.
\label{eq:structural}
\end{equation}

Equation \eqref{eq:structural} makes clear that endogeneity of $Y_2$ arises from the nonzero dependence between $\varepsilon_1$ and $\varepsilon_2$, captured through the unknown function $h(\varepsilon_2)$.

\begin{assumption}\label{ass:baseline}
The following conditions hold.
\begin{enumerate}
    \item\label{ass:exoX}
    The random vector $X$ is mean independent of $\varepsilon_1$ and $\varepsilon_2$: 
    $E[\varepsilon_1 \mid X] = 0$, 
    $E[\varepsilon_2 \mid X] = 0$.
    
    \item\label{ass:cf}
    The disturbance $\eta$ satisfies 
    $E[\eta \mid X, \varepsilon_2] = 0$.

    \item\label{ass:reg}
    The second moments are finite:
    $E\|X\|^2 < \infty$, 
    $E[\varepsilon_2^2] < \infty$, 
    $E[\eta^2] < \infty$.

    \item\label{ass:pd}
    The matrix
    $\Sigma_{XX} := E[XX']$ 
    is positive definite.
\end{enumerate}
\end{assumption}

Assumption~\ref{ass:baseline}.\ref{ass:exoX} imposes mean independence of $X$ from both $\varepsilon_2$ and $\eta$. 
Assumption~\ref{ass:baseline}.\ref{ass:cf} requires that conditional on 
$(X,\varepsilon_2)$, the remaining disturbance $\eta$ has mean zero.  
Thus any systematic dependence between $\varepsilon_1$ and 
$\varepsilon_2$ is captured by $h(\varepsilon_2)$. 
Assumption~\ref{ass:baseline}.\ref{ass:reg} is a standard regularity 
condition guaranteeing the existence of second moments. 
Assumption~\ref{ass:baseline}.\ref{ass:pd} is a basic full-rank condition for identification.

Under \eqref{eq:control-decomp}, we know 
$E[\varepsilon_1 \mid \varepsilon_2] = h(\varepsilon_2)$. 
Therefore, unless $h$ is constant, $\varepsilon_1$ and $\varepsilon_2$ are correlated, implying that $Y_2$ is endogenous in \eqref{eq:secondstage}.

If $h(\cdot)$ were known, one could treat $h(\varepsilon_2)$ as a control function and recover $\gamma_1$ by conditioning on $\varepsilon_2$. However, $h$ is unknown and unrestricted. In the absence of a valid excluded instrument, the structural parameters $(\beta_1,\gamma_1)$ are generally not point identified.

The central problem of this paper is therefore how to extract identifying information about $(\beta_1,\gamma_1)$ in the presence of an unknown nonlinear control function $h(\varepsilon_2)$ and without relying on classical exclusion restrictions. 


\subsection{Baseline Orthogonality}

Assumption~\ref{ass:baseline}.\ref{ass:cf} implies a broad class of orthogonality restrictions.
For any measurable function $f:\mathbb{R}\to\mathbb{R}$ such that
$E|f(\varepsilon_2)\eta|<\infty$, we have
\begin{align*}
E\!\left[f(\varepsilon_2)\eta\right]
&=
E\!\Big[\,E\{f(\varepsilon_2)\eta\mid X,\varepsilon_2\}\,\Big]
=
E\!\Big[\,f(\varepsilon_2)\,E(\eta\mid X,\varepsilon_2)\,\Big]
=
0.
\end{align*}
Using $\varepsilon_1=h(\varepsilon_2)+\eta$ and \eqref{eq:secondstage}, this yields the
\emph{population orthogonality} condition
\begin{equation}
E\!\left[
f(\varepsilon_2)\,
\big\{Y_1 - X'\beta_1 - \gamma_1 Y_2 - h(\varepsilon_2)\big\}
\right]
=0
\qquad\text{for all admissible } f.
\label{eq:orth-pop}
\end{equation}

Equation \eqref{eq:orth-pop} is mechanically implied by the restriction in Assumption~\ref{ass:baseline}.\ref{ass:cf}. In this sense, these moments are ``baseline'': they do not
create identifying power by themselves, because the unknown function $h(\cdot)$
enters inside the moment in a fully unrestricted way. Even taking a large
dictionary $\{f_1,\dots,f_L\}$ does not fundamentally alter this issue; it simply
generates a larger system of moment equations:
\begin{equation}
E\!\left[
f_\ell(\varepsilon_2)\,
\big\{Y_1 - X'\beta_1 - \gamma_1 Y_2 - h(\varepsilon_2)\big\}
\right]
=0,
\qquad \ell=1,\dots,L.
\label{eq:L-moments}
\end{equation}
Because $h(\varepsilon_2)$ is itself an unknown function of $\varepsilon_2$,
multiplying by an additional function $f(\varepsilon_2)$ typically does not
generate qualitatively new information.

The orthogonality conditions in \eqref{eq:orth-pop} involve the latent shock
$\varepsilon_2$. However, in our triangular system, $\varepsilon_2$ is
\emph{proxy-observable} through the first stage.

Under Assumption~\ref{ass:baseline}.\ref{ass:exoX} and  \ref{ass:baseline}.\ref{ass:pd}, $\beta_2$ is point identified from \eqref{eq:firststage},
and the OLS estimator $\hat\beta_2$ is consistent. Define the generated residual
\begin{equation}
\hat\varepsilon_{2i}
:=
Y_{2i}-X_i'\hat\beta_2,
\qquad
\hat\beta_2 \xrightarrow{p} \beta_2,
\qquad
\hat\varepsilon_{2i}\xrightarrow{p}\varepsilon_{2i}.
\label{eq:generated-eps2}
\end{equation}
Thus, any moment condition involving $\varepsilon_2$ can be implemented
asymptotically by replacing $\varepsilon_2$ with $\hat\varepsilon_2$, subject to
standard regularity conditions for generated regressors.

To operationalize \eqref{eq:L-moments}, one must impose some structure on
$h(\cdot)$. A convenient approach is a sieve approximation.
Let $\{\phi_m\}_{m\ge 1}$ be a basis on $\mathbb{R}$ and define the sieve class
\[
\mathcal H_M
:=
\left\{
h_{\bar \theta}(\cdot)=\sum_{m=1}^M \bar \theta_m \phi_m(\cdot)
:\ \bar \theta\in\mathbb{R}^M
\right\}.
\]
For a candidate $(\beta_1,\gamma_1,\bar \theta)$, define the sample moments
\begin{equation}
\hat g_n(\beta_1,\gamma_1,\bar \theta)
=
\frac{1}{n}\sum_{i=1}^n
f(\hat\varepsilon_{2i})\,
\Big\{Y_{1i} - X_i'\beta_1 - \gamma_1 Y_{2i} - h_{\bar \theta}(\hat\varepsilon_{2i})\Big\}.
\label{eq:sample-moments}
\end{equation}
Under Assumption~\ref{ass:baseline}.\ref{ass:exoX}--\ref{ass:baseline}.\ref{ass:pd} and standard conditions (LLN, consistency of $\hat\beta_2$,
boundedness/smoothness of $f$ and $\{\phi_m\}$), we have
\[
\hat g_n(\beta_1,\gamma_1,\bar \theta)
\xrightarrow{p}
g(\beta_1,\gamma_1,\bar \theta)
:=
E\!\left[
f(\varepsilon_2)\,
\{Y_1 - X'\beta_1 - \gamma_1 Y_2 - h_{\bar \theta}(\varepsilon_2)\}
\right].
\]
At the true parameters, taking $h_{\bar \theta}=h$ gives
$g(\beta_1,\gamma_1,\bar \theta)=0$ by \eqref{eq:orth-pop}.

Even after restricting $h$ to a sieve class, $(\beta_1,\gamma_1)$ is typically
not point identified from \eqref{eq:sample-moments}--\eqref{eq:L-moments}.
Intuitively, the unknown function $h(\varepsilon_2)$ absorbs the very component
that induces endogeneity, and without exclusion restrictions there is no
external variation to separately pin down $\gamma_1$ from $h(\cdot)$.
As the sieve dimension $M$ grows, the model becomes increasingly flexible, and
a wider range of $(\beta_1,\gamma_1)$ can be rationalized by some
$h_{\bar \theta}\in\mathcal H_M$.

Although the orthogonality conditions are weak, it is useful to view them through
a standard GMM lens to clarify the nature of the identification problem.

Let $\{f_1,\dots,f_L\}$ be a finite dictionary and stack the corresponding sample
moments as $\hat g_{n,L}(\beta_1,\gamma_1,\bar \theta)\in\mathbb{R}^L$. Let
$\hat\Sigma_{n}$ be a consistent estimator of the asymptotic covariance of
$\sqrt{n}\hat g_{n,L}$. Define the quadratic form
\begin{equation}
J_n(\beta_1,\gamma_1,\bar \theta)
=
n\,\hat g_{n,L}(\beta_1,\gamma_1,\bar \theta)'\,
\hat\Sigma_n^{-1}\,
\hat g_{n,L}(\beta_1,\gamma_1,\bar \theta).
\label{eq:Jstat}
\end{equation}
For fixed $(\beta_1,\gamma_1,\bar \theta)$ satisfying the population moments, $J_n$
has an asymptotic $\chi^2$ distribution under standard regularity. This suggests
an inversion-based \emph{set} for $(\beta_1,\gamma_1)$:
\begin{equation}
\mathcal{C}_{1-\alpha}
=
\Big\{
(\beta_1,\gamma_1):
\ \exists\bar \theta\in\mathbb{R}^M \ \text{s.t.}\
J_n(\beta_1,\gamma_1,\bar \theta)\le \chi^2_{L,1-\alpha}
\Big\}.
\label{eq:conf-set}
\end{equation}

The set \eqref{eq:conf-set} is typically very wide when $h$ is flexible and $M$
is moderately large, reflecting weak identifying content of the baseline
orthogonality. This motivates the need for additional structure that provides
genuine identifying power without relying on classical exclusion restrictions.

In the next section, we introduce such structure by exploiting a cross-group
distributional shift in $\varepsilon_2$ induced by an auxiliary variable $Z_0$,
and by controlling the approximation error of $h$ using sieve arguments.

\section{Identification via Distributional Shift Induced by Instruments}

\subsection{Setup of Instruments}

To generate identifying variation beyond the baseline orthogonality
conditions, we introduce an auxiliary binary variable
\[
Z_0 \in \{0,1\}.
\]

Importantly, $Z_0$ is \emph{not} assumed to satisfy a classical exclusion
restriction for the structural equation \eqref{eq:secondstage}.
Instead, $Z_0$ is allowed to affect the distribution of the first-stage
disturbance $\varepsilon_2$.

Define conditional expectations
\[
E_d[\cdot] := E[\cdot \mid Z_0=d],
\qquad
\Delta E[\cdot] := E_1[\cdot] - E_0[\cdot].
\]

Let
\[
u := \varepsilon_2 - e_0
\]
denote the first-stage disturbance centered around a fixed expansion
point $e_0 \in \mathbb{R}$.

Our identification strategy exploits systematic differences between
the conditional distributions of $u$ across $Z_0=0$ and $Z_0=1$.

Let $f_d(u)$ denote the conditional density of $u$ given $Z_0=d$.
Define the log-density ratio
\begin{equation}
\ell(u)
:=
\log \frac{f_1(u)}{f_0(u)}.
\label{eq:log-density-ratio}
\end{equation}

By construction,
\begin{equation}
f_1(u) = e^{\ell(u)} f_0(u),
\label{eq:density-ratio}
\end{equation}
and therefore for any measurable function $g$ such that the expectations exist,
\begin{equation}
\Delta E[g(u)]
=
E_1[g(u)] - E_0[g(u)]
=
E_0\!\left[g(u)\big(e^{\ell(u)} - 1\big)\right].
\label{eq:deltaE-basic}
\end{equation}

Equation \eqref{eq:deltaE-basic} is an exact identity and forms the
starting point of our identification analysis.

We impose a weak structural restriction on the local shape of the
log-density ratio.

\begin{assumption}[Local quadratic tilting with bounded remainder]
\label{ass:local-tilting}
Fix a bandwidth $b>0$ and define the local window
\[
\mathcal U_b := \{u: |u| \le b\}.
\]
There exist scalars $a\in\mathbb{R}$, $\lambda\in\mathbb{R}$,
and a measurable function $r(\cdot)$ such that
\begin{equation}
\ell(u)
=
a + \lambda u^2 + r(u),
\qquad \forall u\in \mathcal U_b,
\label{eq:lr-decomp}
\end{equation}
and
\begin{equation}
\sup_{u\in\mathcal U_b} |r(u)|
\le \delta_b,
\label{eq:r-bound}
\end{equation}
where $\delta_b \ge 0$ is a sensitivity parameter.
\end{assumption}

Assumption~\ref{ass:local-tilting} does not impose exact exponential
tilting. It only restricts the local curvature of the log-density
ratio up to a uniformly bounded misspecification $\delta_b$.
When $\delta_b=0$, the model reduces locally to exact quadratic tilting.

\begin{remark}
Assumption~\ref{ass:local-tilting} is best understood as a restriction on how the
auxiliary variable $Z_0$ shifts the \emph{shape} of the first-stage disturbance
$u=\varepsilon_2-e_0$ locally around $u=0$. The quadratic component $\lambda u^2$
captures a change in \emph{local dispersion} (or tail thickness) across groups,
while allowing for a bounded local misspecification $r(u)$.

A natural data-generating mechanism is a \emph{state-dependent volatility} model.
Suppose
\[
u=\sigma(Z_0)\,\nu,\qquad E[\nu]=0,
\]
where $\nu$ has a common density $f_\nu$ across $Z_0$ and $\sigma(1)\neq\sigma(0)$.
Then the conditional densities satisfy
\[
f_d(u)=\frac{1}{\sigma(d)}\,f_\nu\!\left(\frac{u}{\sigma(d)}\right),
\qquad d\in\{0,1\},
\]
and hence the log-density ratio is
\begin{equation}\label{eq:ell-scale}
\ell(u)=\log\frac{f_1(u)}{f_0(u)}
=
\log\frac{\sigma(0)}{\sigma(1)}
+
\log f_\nu\!\left(\frac{u}{\sigma(1)}\right)
-
\log f_\nu\!\left(\frac{u}{\sigma(0)}\right).
\end{equation}
If $\log f_\nu(\cdot)$ is twice continuously differentiable at $0$ and $\nu$ is
symmetric (so that $(\log f_\nu)'(0)=0$), a second-order Taylor expansion of
\eqref{eq:ell-scale} around $u=0$ yields
\[
\ell(u)
=
a
+
\lambda u^2
+
r(u),
\qquad |u|\le b,
\]
where $\lambda$ is proportional to the difference between $1/\sigma(1)^2$ and
$1/\sigma(0)^2$, and the remainder $r(u)$ is of order $O(b^3)$ under standard
smoothness of $\log f_\nu$. Thus, locally, the density ratio is well
approximated by a quadratic tilting term even when the global shapes of $f_0$
and $f_1$ are non-Gaussian.

A generic local expansion would allow $\ell(u)=a+\kappa u+\lambda u^2+r(u)$.
The absence of the linear term can be interpreted as a \emph{local balancing}
condition at the expansion point $e_0$: we focus on locations where the two
groups are locally matched in the sense that $\ell'(0)\approx 0$ (hence
$\kappa\approx 0$), so that the dominant local difference is in curvature.
This is consistent with the role of $e_0$ as a tuning parameter: in practice,
we search over a grid of $e_0$ values, and the sensitivity parameter $\delta_b$
captures the extent to which $\ell$ deviates from a purely quadratic shape on
$\mathcal U_b$.

The bound $\sup_{|u|\le b}|r(u)|\le\delta_b$ quantifies local departures from the
idealized ``variance-shift'' benchmark. Economically, larger $\delta_b$ allows
for additional local distortions (e.g., mild skewness differences, higher-order
tail differences, or small local mean shifts not fully absorbed by $X$), while
keeping the identifying content interpretable through the curvature parameter
$\lambda$.
\end{remark}

Since $f_1$ must integrate to one, the density-ratio representation
implies the normalization condition
\begin{equation}
E_0[e^{\ell(u)}] = 1.
\label{eq:normalization}
\end{equation}

Under \eqref{eq:lr-decomp}, this becomes
\[
E_0[e^{a + \lambda u^2 + r(u)}]
=
e^{a} E_0[e^{\lambda u^2 + r(u)}]
=
1.
\]
Hence $a$ is not free but must satisfy
\begin{equation}
e^{-a}
=
E_0[e^{\lambda u^2 + r(u)}].
\label{eq:a-exact}
\end{equation}

Using the bound \eqref{eq:r-bound}, we obtain
\[
e^{-\delta_b} E_0[e^{\lambda u^2}]
\le
E_0[e^{\lambda u^2 + r(u)}]
\le
e^{\delta_b} E_0[e^{\lambda u^2}],
\]
which implies
\begin{equation}
-\mathcal A(\lambda) - \delta_b
\;\le\;
a
\;\le\;
-\mathcal A(\lambda) + \delta_b,
\label{eq:a-interval}
\end{equation}
where
\[
\mathcal A(\lambda)
:=
\log E_0[e^{\lambda u^2}].
\]

Thus the intercept $a$ is restricted to a narrow interval
determined by $(\lambda,\delta_b)$.
In the limit $\delta_b\to 0$, $a$ collapses to the exact
cumulant normalization.

From \eqref{eq:r-bound}, we obtain
\[
e^{-\delta_b}
\le
e^{r(u)}
\le
e^{\delta_b},
\]
and therefore
\begin{equation}
e^{-\delta_b} - 1
\;\le\;
e^{r(u)} - 1
\;\le\;
e^{\delta_b} - 1.
\label{eq:directional-bound}
\end{equation}

Combining \eqref{eq:deltaE-basic} and \eqref{eq:lr-decomp}, we obtain
\begin{align}
\Delta E[g(u)]
&=
E_0\!\left[
g(u)\Big(e^{a} e^{\lambda u^2} e^{r(u)} - 1\Big)
\right]
\notag\\
&=
E_0\!\left[
g(u)\Big(e^{a} e^{\lambda u^2} - 1\Big)
\right]
+
E_0\!\left[
g(u)\, e^{a} e^{\lambda u^2}\big(e^{r(u)} - 1\big)
\right].
\label{eq:deltaE-decomp}
\end{align}

The second term in \eqref{eq:deltaE-decomp} is controlled by
the bound \eqref{eq:directional-bound}, while the first term
depends only on $(a,\lambda)$ and the distribution of $u$
under $Z_0=0$.

\medskip

Equation \eqref{eq:deltaE-decomp} shows that cross-group differences
in moments of $u$ are governed by a low-dimensional curvature parameter
$\lambda$ up to a bounded misspecification $\delta_b$.
In the next section, we show how these cross-group moment shifts
translate into identifying restrictions for the structural parameters
$(\beta_1,\gamma_1)$.

\subsection{FWL Representation and Residualized Equations}

Now we connect the distributional shift in the first-stage disturbance induced by $Z_0$ to restrictions on the structural parameters $(\beta_1,\gamma_1)$.
We proceed in three steps.
First, we remove the contribution of $X$ via the Frisch--Waugh--Lovell (FWL) representation, thereby isolating the endogenous component $\varepsilon_2$.
Second, we approximate the unknown function $h(\cdot)$ locally by a sieve projection, which yields a tractable residual decomposition. 
Third, we combine the sieve residual with the density-ratio identity to obtain moment inequalities.

Let $M_X$ denote the (population) residual-maker that projects out the linear
span of $X$. Specifically, for any square-integrable scalar random variable $V$,
define
\[
M_X V := V - X' \pi_V,
\qquad
\pi_V := \arg\min_{\pi\in\mathbb{R}^k} E\big[(V-X'\pi)^2\big].
\]
Equivalently, $M_X V$ is the population residual from regressing $V$ on $X$.

Applying $M_X$ to the first-stage equation \eqref{eq:firststage} yields
\begin{equation}
\tilde Y_2
:=
M_X Y_2
=
M_X(X'\beta_2+\varepsilon_2)
=
M_X \varepsilon_2.
\label{eq:FWL-firststage}
\end{equation}
Under Assumption~\ref{ass:baseline}.\ref{ass:exoX}, $E[\varepsilon_2\mid X]=0$, hence
$\varepsilon_2$ is orthogonal to $X$ and the linear projection of $\varepsilon_2$
on $X$ is zero. Therefore,
\begin{equation}
\tilde Y_2 = \varepsilon_2.
\label{eq:tildeY2-eps2}
\end{equation}

Next apply $M_X$ to the structural equation \eqref{eq:structural}:
\begin{align}
\tilde Y_1
:=
M_X Y_1
&=
M_X\big(X'\beta_1+\gamma_1 Y_2+h(\varepsilon_2)+\eta\big)
\notag\\
&=
\gamma_1 M_X Y_2 + M_X h(\varepsilon_2) + M_X \eta.
\label{eq:FWL-structural-raw}
\end{align}
Under Assumption~\ref{ass:baseline}.\ref{ass:exoX}, $E[\eta\mid X]=0$, so $M_X\eta=\eta$.
Moreover, since $E[h(\varepsilon_2)\mid X]=0$ follows from
$E[\varepsilon_2\mid X]=0$ and the control-function structure Assumption~\ref{ass:baseline}.\ref{ass:cf},
we have $M_X h(\varepsilon_2)=h(\varepsilon_2)$.\footnote{Formally, under Assumption~\ref{ass:baseline}.\ref{ass:cf},
$E[\eta\mid X,\varepsilon_2]=0$ implies $E[\varepsilon_1\mid X,\varepsilon_2]=h(\varepsilon_2)$.
Together with $E[\varepsilon_1\mid X]=0$ (exogeneity of $X$), this yields
$E[h(\varepsilon_2)\mid X]=0$.}
Using \eqref{eq:tildeY2-eps2}, we obtain the residualized structural equation
\begin{equation}
\tilde Y_1
=
\gamma_1 \varepsilon_2 + h(\varepsilon_2) + \eta.
\label{eq:tildeY1-rep}
\end{equation}

Fix a localization point $e_0\in\mathbb{R}$ and define
\[
u:=\varepsilon_2-e_0.
\]
Let $\psi_b(u)=K(u/b)\mathbf 1\{|u|\le b\}$ be a bounded local weight with
bandwidth $b>0$.

\subsection{Local Sieve Approximation}

We approximate $h(\varepsilon_2)$ locally by a sieve projection under the
baseline group $Z_0=0$. Let $\mu_{0,b}$ denote the distribution of
$u=\varepsilon_2-e_0$ conditional on $Z_0=0$ restricted to the window
$\{|u|\le b\}$.

Let $\{p_k(u)\}_{k\ge1}$ be a complete basis for $L^2(\mu_{0,b})$.
For each $j\ge1$, define the sieve space
\[
\mathcal H_j
:=
\left\{
h_j(u)=\sum_{k=1}^j \alpha_k p_k(u)
\right\}.
\]
Define $h_j$ as the weighted $L^2(\mu_{0,b})$-projection of $h$ onto
$\mathcal H_j$:
\begin{equation}
h_j
=
\arg\min_{g\in \mathcal H_j}
E_0\!\left[\psi_b(u)\big(h(u)-g(u)\big)^2\right].
\label{eq:sieve-proj}
\end{equation}
Let the approximation error be
\[
r_j(u):=h(u)-h_j(u).
\]

\begin{assumption}[Local sieve approximation error]
\label{ass:sieve-error}
There exists a nonnegative sequence $\delta_h(j)$ such that
\begin{equation}
E_0\!\left[\psi_b(u)\,r_j(u)^2\right]
\le
\delta_h(j)^2,
\qquad
\delta_h(j)\to0 \ \ \text{as } j\to\infty.
\label{eq:sieve-rate}
\end{equation}
\end{assumption}

Assumption~\ref{ass:sieve-error} is a high-level restriction on the
local complexity of $h$ relative to the sieve basis and the window $b$.
It does not impose parametric structure on $h$.

Define the sieve-adjusted outcome
\begin{equation}
\tilde W_j
:=
\tilde Y_1 - \gamma_1 \tilde Y_2 - h_j(\tilde Y_2).
\label{eq:Wj-def}
\end{equation}
Substituting \eqref{eq:tildeY2-eps2} and \eqref{eq:tildeY1-rep} into
\eqref{eq:Wj-def} yields
\begin{align}
\tilde W_j
&=
\big(\gamma_1 \varepsilon_2 + h(\varepsilon_2) + \eta\big)
-
\gamma_1 \varepsilon_2
-
h_j(\varepsilon_2)
\notag\\
&=
\eta + \big(h(\varepsilon_2)-h_j(\varepsilon_2)\big)
=
\eta + r_j(u).
\label{eq:Wj-decomp}
\end{align}
Thus, after subtracting the sieve approximation of $h$,
the remaining nontrivial component is the sieve approximation error $r_j(u)$.

Recall that for any measurable function $g$,
\begin{equation}
\Delta E[g(u)]
=
E_1[g(u)]-E_0[g(u)]
=
E_0\!\left[g(u)\big(e^{\ell(u)}-1\big)\right],
\label{eq:deltaE-basic-repeat}
\end{equation}
where $\ell(u)=\log(f_1(u)/f_0(u))$ denotes the log-density ratio of $u$ under
$Z_0=1$ relative to $Z_0=0$.

We choose the test function
\begin{equation}
g(u):=\psi_b(u)\tilde W_j.
\label{eq:g-choice}
\end{equation}
Substituting \eqref{eq:g-choice} into \eqref{eq:deltaE-basic-repeat} yields the
exact identity
\begin{equation}
\Delta E[\psi_b(u)\tilde W_j]
=
E_0\!\left[\psi_b(u)\tilde W_j\big(e^{\ell(u)}-1\big)\right].
\label{eq:key-identity}
\end{equation}

Using the decomposition \eqref{eq:Wj-decomp}, the right-hand side becomes
\begin{align}
E_0\!\left[\psi_b(u)\tilde W_j\big(e^{\ell(u)}-1\big)\right]
&=
E_0\!\left[\psi_b(u)\big(\eta+r_j(u)\big)\big(e^{\ell(u)}-1\big)\right]
\notag\\
&=
E_0\!\left[\psi_b(u)\eta\big(e^{\ell(u)}-1\big)\right]
+
E_0\!\left[\psi_b(u)r_j(u)\big(e^{\ell(u)}-1\big)\right].
\label{eq:split-eta-rj}
\end{align}

Since $e^{\ell(u)}-1$ and $\psi_b(u)$ are measurable functions of $u$, Assumption~\ref{ass:baseline}.\ref{ass:cf} implies
\begin{align*}
E_0\!\left[\psi_b(u)\eta\big(e^{\ell(u)}-1\big)\right]
&=
E_0\!\Big[
E\{\psi_b(u)\eta(e^{\ell(u)}-1)\mid X,\varepsilon_2,Z_0\}
\Big]\\
&=
E_0\!\Big[
\psi_b(u)(e^{\ell(u)}-1)\,E(\eta\mid X,\varepsilon_2)
\Big]
=
0.
\end{align*}
Therefore \eqref{eq:key-identity} and \eqref{eq:split-eta-rj} imply
\begin{equation}
\Delta E[\psi_b(u)\tilde W_j]
=
E_0\!\left[\psi_b(u)r_j(u)\big(e^{\ell(u)}-1\big)\right].
\label{eq:remainder-only}
\end{equation}

Equation \eqref{eq:remainder-only} is the key link between the observable
cross-group difference $\Delta E[\psi_b\tilde W_j]$ and the unknown sieve
approximation error $r_j(u)$.

\subsection{Moment Inequalities}

We now bound the right-hand side of \eqref{eq:remainder-only}.
By Cauchy--Schwarz,
\begin{align}
\left|
E_0\!\left[\psi_b(u)r_j(u)\big(e^{\ell(u)}-1\big)\right]
\right|
&\le
\left(
E_0[\psi_b(u)r_j(u)^2]
\right)^{1/2}
\left(
E_0[\psi_b(u)\big(e^{\ell(u)}-1\big)^2]
\right)^{1/2}.
\label{eq:CS}
\end{align}

By Assumption~\ref{ass:sieve-error},
\[
E_0[\psi_b(u)r_j(u)^2]
\le
\delta_h(j)^2.
\]
Define
\begin{equation}
V_b
:=
E_0\!\left[\psi_b(u)\big(e^{\ell(u)}-1\big)^2\right].
\label{eq:Vb-def}
\end{equation}
Combining \eqref{eq:remainder-only}--\eqref{eq:CS} yields the bound
\begin{equation}
\left|
\Delta E[\psi_b(u)\tilde W_j]
\right|
\le
\delta_h(j)\sqrt{V_b}.
\label{eq:main-bound}
\end{equation}

Finally, we express \eqref{eq:main-bound} as two one-sided moment inequalities:
\begin{align}
\Delta E[\psi_b(u)\tilde W_j] - \delta_h(j)\sqrt{V_b}
&\le 0,
\label{eq:ineq-1}\\
-\Delta E[\psi_b(u)\tilde W_j] - \delta_h(j)\sqrt{V_b}
&\le 0.
\label{eq:ineq-2}
\end{align}

Equivalently, defining
\[
m_{1,j,b}(\theta)
:=
\delta_h(j)\sqrt{V_b}-\Delta E[\psi_b(u)\tilde W_j(\theta)],
\qquad
m_{2,j,b}(\theta)
:=
\delta_h(j)\sqrt{V_b}+\Delta E[\psi_b(u)\tilde W_j(\theta)],
\]
where $\theta$ collects the structural parameters $(\beta_1,\gamma_1)$
(and any nuisance parameters used in $h_j$),
the model implies the pair of moment inequalities
\[
m_{1,j,b}(\theta)\ge 0,
\qquad
m_{2,j,b}(\theta)\ge 0.
\]

The key identifying content in \eqref{eq:ineq-1}--\eqref{eq:ineq-2} arises from
two sources: (i) the distributional shift in $u$ across $Z_0$ encoded in
$\ell(u)$, and (ii) the shrinking sieve approximation error $\delta_h(j)$ as the
sieve dimension grows. Therefore, we are able to study how these inequalities restrict
the structural parameters and how their informativeness varies with $(j,b,e_0)$.

\begin{corollary}[Point identification of $\gamma_1$ as sieve bias vanishes]
\label{cor:point-id-gamma}

Maintain Assumption~\ref{ass:baseline} and the density-ratio setup, including the identity \eqref{eq:deltaE-basic-repeat}. 
Fix $(b,e_0)$ and define $u=\varepsilon_2-e_0$ and the local weight
$\psi_b(u)=K(u/b)\mathbf 1\{|u|\le b\}$.

For each sieve dimension $j$, let $h_j$ be the local $L^2(\mu_{0,b})$-projection
of $h$ defined in \eqref{eq:sieve-proj}, and let $r_j(u)=h(u)-h_j(u)$.
For any $\gamma\in\mathbb{R}$, define the population residualized object
\begin{equation}
W_j(\gamma)
:=
\tilde Y_1-\gamma \tilde Y_2-h_j(\tilde Y_2),
\label{eq:Wj-gamma-def}
\end{equation}
where $\tilde Y_1=M_XY_1$ and $\tilde Y_2=M_XY_2$ are the population FWL residuals.
Let the $j$-th identified set for $\gamma$ be
\begin{equation}
\Gamma_j
:=
\left\{
\gamma\in\mathbb{R}:
\left|\Delta E\!\left[\psi_b(u)\,W_j(\gamma)\right]\right|
\le
\delta_h(j)\sqrt{V_b}
\right\},
\label{eq:Gammaj-def}
\end{equation}
where $V_b:=E_0[\psi_b(u)(e^{\ell(u)}-1)^2]$ as in \eqref{eq:Vb-def}, and
$\delta_h(j)$ satisfies Assumption~\ref{ass:sieve-error}.

Assume the following \emph{local relevance} condition holds:
\begin{equation}
\Delta E\!\left[\psi_b(u)\,\varepsilon_2\right]\neq 0,
\label{eq:local-relevance}
\end{equation}
and $0<V_b<\infty$. Then:
\begin{enumerate}
\item[(i)] For each $j$, the true coefficient $\gamma_1$ belongs to $\Gamma_j$.
\item[(ii)] The diameter of $\Gamma_j$ satisfies the bound
\begin{equation}
\mathrm{diam}(\Gamma_j)
\le
\frac{2\,\delta_h(j)\sqrt{V_b}}{\left|\Delta E[\psi_b(u)\varepsilon_2]\right|}.
\label{eq:Gammaj-diam}
\end{equation}
\item[(iii)] In particular, if $\delta_h(j)\to 0$ as $j\to\infty$, then
$\mathrm{diam}(\Gamma_j)\to 0$ and hence $\Gamma_j$ shrinks to the singleton
$\{\gamma_1\}$. Equivalently, $\gamma_1$ is point identified in the limit
$j\to\infty$ by the family of restrictions \eqref{eq:Gammaj-def}.
\end{enumerate}
\end{corollary}

\begin{proof}
Under Assumption~\ref{ass:baseline}.\ref{ass:exoX}, the population FWL residuals
satisfy $\tilde Y_2=\varepsilon_2$ and
$\tilde Y_1=\gamma_1\varepsilon_2+h(\varepsilon_2)+\eta$
as shown in \eqref{eq:tildeY1-rep}.
Substituting into \eqref{eq:Wj-gamma-def} yields, for any $\gamma$,
\begin{align}
W_j(\gamma)
&=
\big(\gamma_1\varepsilon_2+h(\varepsilon_2)+\eta\big)-\gamma\varepsilon_2-h_j(\varepsilon_2)
\notag\\
&=
(\gamma_1-\gamma)\varepsilon_2+\eta+r_j(u).
\label{eq:Wj-gamma-decomp}
\end{align}

Because $\psi_b(u)$ is measurable with respect to $\varepsilon_2$ (hence
measurable with respect to $(X,\varepsilon_2,Z_0)$), we have
\[
E\!\left[\psi_b(u)\eta\mid Z_0\right]
=
E\!\left[E\!\left[\psi_b(u)\eta\mid X,\varepsilon_2,Z_0\right]\mid Z_0\right]
=
E\!\left[\psi_b(u)\,E[\eta\mid X,\varepsilon_2,Z_0]\mid Z_0\right].
\]
Assumption~\ref{ass:baseline}.\ref{ass:cf} gives $E[\eta\mid X,\varepsilon_2]=0$
almost surely. Hence
$E[\eta\mid X,\varepsilon_2,Z_0]
=
E\!\left[E[\eta\mid X,\varepsilon_2]\mid X,\varepsilon_2,Z_0\right]
=0$ almost surely, implying $E[\psi_b(u)\eta\mid Z_0]=0$ for each $Z_0$.
Therefore
\begin{equation}
\Delta E[\psi_b(u)\eta]=E_1[\psi_b(u)\eta]-E_0[\psi_b(u)\eta]=0.
\label{eq:DeltaE-eta-zero}
\end{equation}

Taking $\Delta E[\psi_b(\cdot)\cdot]$ on both sides of
\eqref{eq:Wj-gamma-decomp} and using \eqref{eq:DeltaE-eta-zero} gives
\begin{equation}
\Delta E[\psi_b(u)W_j(\gamma)]
=
(\gamma_1-\gamma)\Delta E[\psi_b(u)\varepsilon_2]
+
\Delta E[\psi_b(u)r_j(u)].
\label{eq:DeltaE-Wj-linear}
\end{equation}

Next, we bound the remainder term $\Delta E[\psi_b(u)r_j(u)]$.
By the density-ratio identity \eqref{eq:deltaE-basic-repeat},
\[
\Delta E[\psi_b(u)r_j(u)]
=
E_0\!\left[\psi_b(u)r_j(u)(e^{\ell(u)}-1)\right].
\]
By Cauchy--Schwarz and the definition of $V_b$ in \eqref{eq:Vb-def},
\begin{align}
\left|\Delta E[\psi_b(u)r_j(u)]\right|
&\le
\left(E_0[\psi_b(u)r_j(u)^2]\right)^{1/2}
\left(E_0[\psi_b(u)(e^{\ell(u)}-1)^2]\right)^{1/2}
\notag\\
&=
\left(E_0[\psi_b(u)r_j(u)^2]\right)^{1/2}\sqrt{V_b}
\notag\\
&\le
\delta_h(j)\sqrt{V_b},
\label{eq:DeltaE-rj-bound}
\end{align}
where the last inequality uses Assumption~\ref{ass:sieve-error}.

Combining \eqref{eq:DeltaE-Wj-linear} and \eqref{eq:DeltaE-rj-bound}, we obtain
\begin{equation}
\left|
\Delta E[\psi_b(u)W_j(\gamma)]
-
(\gamma_1-\gamma)\Delta E[\psi_b(u)\varepsilon_2]
\right|
\le
\delta_h(j)\sqrt{V_b}.
\label{eq:key-ineq-gamma}
\end{equation}

Setting $\gamma=\gamma_1$ in \eqref{eq:key-ineq-gamma} gives
\[
\left|\Delta E[\psi_b(u)W_j(\gamma_1)]\right|
\le
\delta_h(j)\sqrt{V_b},
\]
which is exactly $\gamma_1\in\Gamma_j$ by \eqref{eq:Gammaj-def}.

Take any $\gamma\in\Gamma_j$. By definition \eqref{eq:Gammaj-def},
$\left|\Delta E[\psi_b(u)W_j(\gamma)]\right|\le \delta_h(j)\sqrt{V_b}$.
Plugging this into \eqref{eq:key-ineq-gamma} and applying the triangle inequality,
\[
\left|(\gamma_1-\gamma)\Delta E[\psi_b(u)\varepsilon_2]\right|
\le
\left|\Delta E[\psi_b(u)W_j(\gamma)]\right|+\delta_h(j)\sqrt{V_b}
\le
2\delta_h(j)\sqrt{V_b}.
\]
Under the local relevance condition \eqref{eq:local-relevance}, we can divide to obtain
\[
|\gamma-\gamma_1|
\le
\frac{2\,\delta_h(j)\sqrt{V_b}}{|\Delta E[\psi_b(u)\varepsilon_2]|}.
\]
This implies \eqref{eq:Gammaj-diam}.

If $\delta_h(j)\to 0$, then \eqref{eq:Gammaj-diam} implies
$\mathrm{diam}(\Gamma_j)\to 0$. Since $\gamma_1\in\Gamma_j$ for all $j$, it follows that
$\Gamma_j$ converges to the singleton $\{\gamma_1\}$.
\end{proof}

\begin{remark}
Condition \eqref{eq:local-relevance} is a local relevance requirement.
In practice, one may consider a grid of $(b,e_0)$ values and intersect the
resulting sets; if \eqref{eq:local-relevance} fails at some $(b,e_0)$, those
choices contribute no identifying power but do not invalidate the procedure.
\end{remark}

\section{Implementation and Inference by Test Inversion}
\label{sec:implementation}

This section describes how to operationalize the moment inequalities using sample analogues and how to construct confidence sets for
$(\beta_1,\gamma_1)$ by test inversion.  We separate (i) construction of
generated regressors and residualized objects, (ii) estimation and calibration
of nuisance quantities in the bounds, and (iii) the inversion procedure.

Throughout, let $\theta:=(\beta_1,\gamma_1)$ denote the structural parameters of
interest. The implementation depends on tuning parameters
\[
(j,b,e_0)\in\mathcal J\times\mathcal B\times\mathcal E,
\]
where $j$ is the sieve dimension, $b$ is the localization bandwidth, and $e_0$ is
the localization point. The final confidence set will intersect across a
user-chosen grid of $(j,b,e_0)$ values.

\subsection{Generated first-stage residuals and FWL residualization}

Estimate \eqref{eq:firststage} by OLS to obtain $\hat\beta_2$ and define the
generated first-stage residual
\begin{equation}
\hat\varepsilon_{2i}
:=
Y_{2i}-X_i'\hat\beta_2,
\qquad
\hat u_i
:=
\hat\varepsilon_{2i}-e_0.
\label{eq:u-hat}
\end{equation}
Under Assumptions Assumption~\ref{ass:baseline}.\ref{ass:exoX}--Assumption~\ref{ass:baseline}.\ref{ass:pd}, $\hat\beta_2\xrightarrow{p}\beta_2$ and
$\hat u_i$ is a consistent proxy for $u_i=\varepsilon_{2i}-e_0$.

Let $\hat M_X$ denote the sample residual-maker onto the orthogonal complement
of the columns of $X$. Define
\begin{equation}
\tilde Y_{1i}:=\hat M_X Y_{1i},\qquad
\tilde Y_{2i}:=\hat M_X Y_{2i}.
\label{eq:tildeY-sample}
\end{equation}
By the sample FWL theorem, $\tilde Y_{2i}$ is the residual of $Y_{2i}$ after
partialling out $X_i$ and similarly for $\tilde Y_{1i}$.

\medskip

In the population, $\tilde Y_2=\varepsilon_2$ and
$\tilde Y_1=\gamma_1\varepsilon_2+h(\varepsilon_2)+\eta$.
In implementation, we treat $\tilde Y_{2i}$ as a proxy for $\varepsilon_{2i}$,
and use $\hat u_i$ as a proxy for $u_i=\varepsilon_{2i}-e_0$.

\subsection{Sieve construction and the sieve-adjusted residual}

Fix $(j,b,e_0)$. Let $\{p_k(\cdot)\}_{k=1}^j$ be a chosen basis
(e.g., spline or polynomial) and define the sieve class
\[
\mathcal H_j=\Big\{\sum_{k=1}^j \alpha_k p_k(\cdot)\Big\}.
\]
We approximate $h$ by $h_j\in\mathcal H_j$ using only the $Z_0=0$ subsample and
local weight $\psi_b(\cdot)$. Let
\begin{equation}
\psi_b(t):=K(t/b)\mathbf 1\{|t|\le b\},
\label{eq:psi-def}
\end{equation}
where $K$ is a bounded kernel.

For a given candidate $\gamma_1$ (hence candidate $\theta$),
define the pseudo-outcome
\[
R_{1i}(\gamma_1):=\tilde Y_{1i}-\gamma_1 \tilde Y_{2i}.
\]
Estimate $\alpha(\gamma_1)\in\mathbb{R}^j$ by weighted least squares on the
$Z_0=0$ subsample:
\begin{equation}
\hat\alpha(\gamma_1)
:=
\arg\min_{\alpha\in\mathbb{R}^j}
\sum_{i:Z_{0i}=0}
\psi_b(\hat u_i)\Big(R_{1i}(\gamma_1)-\sum_{k=1}^j \alpha_k p_k(\hat u_i)\Big)^2.
\label{eq:alpha-hat}
\end{equation}
Set the fitted sieve function
\[
\hat h_j(\cdot;\gamma_1):=\sum_{k=1}^j \hat\alpha_k(\gamma_1)\,p_k(\cdot).
\]

Define the sample analogue of $\tilde W_j$ as
\begin{equation}
\hat W_{j,i}(\theta)
:=
\tilde Y_{1i}
-
\gamma_1 \tilde Y_{2i}
-
\hat h_j(\hat u_i;\gamma_1).
\label{eq:W-hat}
\end{equation}
This corresponds to the population object
$\tilde W_j=\eta+r_j(u)$.

To mitigate overfitting bias when $j$ is moderately large, one may employ
$K$-fold cross-fitting:
estimate $\hat h_j$ on the training folds within $Z_0=0$ and evaluate
$\hat W_{j,i}$ on the held-out fold, then aggregate across folds.
All subsequent steps can be carried out using the cross-fitted residuals.

\subsection{Estimating the density-ratio nuisance}

The moment bound involves
\[
V_b
=
E_0\!\left[\psi_b(u)\big(e^{\ell(u)}-1\big)^2\right],
\qquad
\ell(u)=\log\{f_1(u)/f_0(u)\}.
\]
We propose a low-dimensional working model for the \emph{local} log-density ratio
based on Assumption~\ref{ass:local-tilting}.

Fix $(b,e_0)$. On the local window $\{i:|\hat u_i|\le b\}$, estimate
$(a,\lambda)$ by fitting the exponential-tilting density ratio
\begin{equation}
\ell(u)\approx a+\lambda u^2
\label{eq:tilt-working}
\end{equation}
between the empirical distributions of $\hat u$ under $Z_0=1$ and $Z_0=0$.
One implementation is \emph{local logistic regression}:
regress $Z_0$ on $1$ and $\hat u^2$ using only observations with
$|\hat u|\le b$ and optionally kernel weights $K(\hat u/b)$.
Let $(\hat a,\hat\lambda)$ denote the resulting coefficients.

Optionally, enforce the normalization $E_0[e^{\ell(u)}]=1$ locally by replacing
$\hat a$ with
\begin{equation}
\hat a^{\,\mathrm{norm}}
:=
-\log\Bigg\{
\frac{1}{n_0}\sum_{i:Z_{0i}=0}\psi_b(\hat u_i)\exp(\hat\lambda \hat u_i^2)
\Bigg\},
\label{eq:a-norm}
\end{equation}
and set $\hat\ell(\hat u_i):=\hat a^{\,\mathrm{norm}}+\hat\lambda \hat u_i^2$.
This corresponds to the exact-tilting normalization when the remainder is zero
and provides a stabilized plug-in in finite samples.

Define
\begin{equation}
\widehat V_b
:=
\frac{1}{n_0}\sum_{i:Z_{0i}=0}
\psi_b(\hat u_i)\big(\exp\{\hat\ell(\hat u_i)\}-1\big)^2.
\label{eq:Vb-hat}
\end{equation}

Under Assumption~\ref{ass:local-tilting}, the remainder $r(\cdot)$ is bounded by
$\delta_b$ on $\mathcal U_b$. In practice, we treat $\delta_b$ as a sensitivity
parameter and report results over a grid $\delta_b\in\mathcal D$.
When one is willing to impose $\delta_b=0$ (exact local tilting), the procedure
simplifies by setting $r\equiv 0$.

\subsection{Calibration of the sieve approximation error}

The bound also depends on $\delta_h(j)$ in Assumption~\ref{ass:sieve-error}. 
To calibrate $\delta_h$, a possible thought is to use a function-class. Specifically, 
assume $h$ belongs to a Sobolev-type ball $\mathcal H(s,R)$ with smoothness
$s>0$ and radius $R>0$. For standard spline/polynomial sieves,
\begin{equation}
\delta_h(j)\le C(s,b)\,R\,j^{-s}.
\label{eq:delta-h-sobolev}
\end{equation}
We treat $(s,R)$ as sensitivity parameters and report results over a grid.

Alternatively, we can use data-driven approach for the calibration. However, it may usually generate a conservative bound. 
Define a conservative bound using cross-fitted prediction error:
\begin{equation}
\widehat{\delta}_h(j)
:=
\Bigg(
\frac{1}{n_0}\sum_{i:Z_{0i}=0}
\psi_b(\hat u_i)\,\hat W_{j,i}(\theta)^2
\Bigg)^{1/2}.
\label{eq:delta-h-hat}
\end{equation}
Since $\hat W_{j,i}(\theta)\approx \eta_i+r_j(u_i)$, this absorbs the noise
variance and hence overestimates the approximation error component, yielding a
conservative procedure.

\subsection{Test inversion and confidence sets}

For fixed $(j,b,e_0)$ and candidate $\theta=(\beta_1,\gamma_1)$, define the sample
analogue of $\Delta E[\psi_b \tilde W_j]$ by
\begin{equation}
\widehat{\Delta E}_n(\theta)
:=
\frac{1}{n_1}\sum_{i:Z_{0i}=1}
\psi_b(\hat u_i)\,\hat W_{j,i}(\theta)
-
\frac{1}{n_0}\sum_{i:Z_{0i}=0}
\psi_b(\hat u_i)\,\hat W_{j,i}(\theta),
\label{eq:DeltaE-hat}
\end{equation}
where $n_d=\sum_{i=1}^n\mathbf 1\{Z_{0i}=d\}$.

Let $\widehat V_b$ be as in \eqref{eq:Vb-hat} and let $\delta_h(j)$ denote a
chosen bound (either \eqref{eq:delta-h-sobolev} or \eqref{eq:delta-h-hat}).
The population restriction \eqref{eq:main-bound} motivates the sample analogue
\begin{equation}
\left|
\widehat{\Delta E}_n(\theta)
\right|
\le
\delta_h(j)\sqrt{\widehat V_b}
+
\tau_{n,j,b},
\label{eq:sample-ineq}
\end{equation}
where $\tau_{n,j,b}\ge 0$ is a (small) statistical tolerance accounting for
sampling uncertainty and generated-regressor effects.\footnote{In practice
$\tau_{n,j,b}$ can be constructed from an asymptotic standard error or from a
multiplier bootstrap. We provide details below.}

Equivalently, \eqref{eq:sample-ineq} is the pair of one-sided inequalities
\begin{align}
\widehat{\Delta E}_n(\theta)-\delta_h(j)\sqrt{\widehat V_b}-\tau_{n,j,b}
&\le 0,
\label{eq:sample-ineq-1}\\
-\widehat{\Delta E}_n(\theta)-\delta_h(j)\sqrt{\widehat V_b}-\tau_{n,j,b}
&\le 0.
\label{eq:sample-ineq-2}
\end{align}

We construct confidence sets for $\theta$ by inverting a test of the null that
the moment inequalities \eqref{eq:sample-ineq-1}--\eqref{eq:sample-ineq-2}
hold for a given candidate $\theta$.

Define the moment vector
\[
\hat m_{n}(\theta)
:=
\begin{pmatrix}
\widehat{\Delta E}_n(\theta)-\delta_h(j)\sqrt{\widehat V_b}\\
-\widehat{\Delta E}_n(\theta)-\delta_h(j)\sqrt{\widehat V_b}
\end{pmatrix}.
\]

A simple choice is the maximum (intersection-union) statistic
\begin{equation}
T_{n}(\theta)
:=
\sqrt{n}\,
\max\big\{
\hat m_{n,1}(\theta),\,
\hat m_{n,2}(\theta),\,
0
\big\}.
\label{eq:test-stat}
\end{equation}
Large values of $T_n(\theta)$ indicate violation of at least one inequality.

Let $\xi_i$ be i.i.d.\ standard normal multipliers independent of the data.
Let $\hat\varphi_i(\theta)$ denote an influence-function-type score for
$\widehat{\Delta E}_n(\theta)$ (constructed from the sample analogues in
\eqref{eq:DeltaE-hat} with cross-fitting).
Define the bootstrap analogue
\[
\widehat{\Delta E}_n^*(\theta)
:=
\frac{1}{\sqrt{n}}\sum_{i=1}^n \xi_i \hat\varphi_i(\theta).
\]
Let
\[
T_n^*(\theta)
:=
\max\{\widehat{\Delta E}_n^*(\theta),-\widehat{\Delta E}_n^*(\theta),0\}.
\]
Let $c_{1-\alpha_G}(\theta)$ be the $(1-\alpha_G)$ quantile of $T_n^*(\theta)$ conditional on the data, where $\alpha_G = \alpha / | \mathcal G|$, and $\mathcal G$ is a user-specified finite grid.

Define
\begin{equation}
\mathcal C_{1-\alpha_G}(j,b,e_0)
:=
\left\{
\theta:
T_n(\theta)\le c_{1-\alpha_G}(\theta)
\right\}.
\label{eq:conf-fixed}
\end{equation}

Finally, define the reported confidence set as the intersection across a grid of
tuning parameters:
\begin{equation}
\mathcal C_{1-\alpha}
:=
\bigcap_{(j,b,e_0)\in\mathcal G}
\mathcal C_{1-\alpha_G}(j,b,e_0),
\label{eq:conf-intersection}
\end{equation}
Intersecting across $(j,b,e_0)$ strengthens robustness and reflects the
multi-scale nature of the identifying restrictions.

We now formalize the asymptotic validity of the test inversion procedure and the
resulting confidence set. The key difficulty is that the moment inequalities are
constructed from (i) generated first-stage residuals, (ii) a sieve regression
estimated on the $Z_0=0$ subsample, and (iii) a locally estimated density-ratio
nuisance. We use cross-fitting to mitigate overfitting bias and employ a
multiplier bootstrap to calibrate the critical value.

\begin{assumption}[Regularity for inference]\label{ass:inference}
The following conditions hold.
\begin{enumerate}
    \item\label{ass:inf-iid}
    $\{(Y_{1i},Y_{2i},X_i,Z_{0i})\}_{i=1}^n$ are i.i.d., with
    $P(Z_0=1)\in(0,1)$ and $E\|X\|^{2+\epsilon}<\infty$ for some $\epsilon>0$.
    
    \item\label{ass:inf-nondeg}
    For each fixed $(j,b,e_0)\in\mathcal G$ and each candidate $\theta$ in a compact
    parameter space $\Theta$, the variance of the influence function for
    $\widehat{\Delta E}_n(\theta)$ is uniformly bounded away from zero and infinity.
    
    \item\label{ass:inf-cf}
    Cross-fitting is implemented so that, conditional on the training folds,
    the evaluation-fold residuals $\hat W_{j,i}(\theta)$ satisfy a uniform
    Neyman-orthogonality / debiasing property:
    the first-order impact of nuisance estimation errors on
    $\widehat{\Delta E}_n(\theta)$ is $o_p(n^{-1/2})$ uniformly in
    $\theta\in\Theta$ and $(j,b,e_0)\in\mathcal G$.
    
    \item\label{ass:inf-rate}
    The nuisance estimators used to construct $\hat W_{j,i}(\theta)$ and
    $\hat\ell(\hat u_i)$ achieve rates sufficient for
    $\sup_{\theta\in\Theta,(j,b,e_0)\in\mathcal G}
    \big|\widehat{\Delta E}_n(\theta)-\Delta E[\psi_b(u)\tilde W_j(\theta)]\big|
    =O_p(n^{-1/2})$.
    Moreover, $\widehat V_b \xrightarrow{p} V_b$ for each $(b,e_0)\in\mathcal G$.
    
    \item\label{ass:inf-bootstrap}
    The multiplier bootstrap statistic $T_n^*(\theta)$ consistently estimates the
    conditional law of $T_n(\theta)$ under the null uniformly over
    $\theta\in\Theta$ and $(j,b,e_0)\in\mathcal G$.
\end{enumerate}
\end{assumption}

\begin{theorem}[Asymptotic validity and coverage]\label{thm:main-inference}
Let $\mathcal G$ be a \emph{finite} grid of tuning parameters $(j,b,e_0)$.
Define $\alpha_G:=\alpha/|\mathcal G|$.
For each $(j,b,e_0)\in\mathcal G$, construct $\mathcal C_{1-\alpha_G}(j,b,e_0)$
as in \eqref{eq:conf-fixed} but using the $(1-\alpha_G)$ bootstrap quantile
$c_{1-\alpha_G}(\theta)$.
Let
\[
\mathcal C_{1-\alpha}
:=
\bigcap_{(j,b,e_0)\in\mathcal G}\mathcal C_{1-\alpha_G}(j,b,e_0).
\]
Under Assumptions~\ref{ass:baseline}, \ref{ass:local-tilting},
\ref{ass:sieve-error}, and \ref{ass:inference}, the confidence set
$\mathcal C_{1-\alpha}$ satisfies the asymptotic coverage guarantee
\[
\liminf_{n\to\infty}P\big(\theta_0\in\mathcal C_{1-\alpha}\big)\ \ge\ 1-\alpha,
\]
where $\theta_0=(\beta_{1,0},\gamma_{1,0})$ is the true structural parameter.
\end{theorem}

\begin{proof}
Fix any $(j,b,e_0)\in\mathcal G$.
Write the corresponding population inequality in \eqref{eq:main-bound} as
\begin{equation}\label{eq:null-fixed}
\big|\Delta E[\psi_b(u)\tilde W_j(\theta_0)]\big|\ \le\ \delta_h(j)\sqrt{V_b},
\end{equation}
which holds by Assumptions~\ref{ass:baseline}.\ref{ass:cf} and \ref{ass:sieve-error}
together with the derivation leading to \eqref{eq:main-bound}.

Define the two population moments
\[
m_1(\theta):=\Delta E[\psi_b(u)\tilde W_j(\theta)]-\delta_h(j)\sqrt{V_b},\qquad
m_2(\theta):=-\Delta E[\psi_b(u)\tilde W_j(\theta)]-\delta_h(j)\sqrt{V_b}.
\]
Then \eqref{eq:null-fixed} is equivalent to $m_1(\theta_0)\le 0$ and
$m_2(\theta_0)\le 0$. Hence,
\begin{equation}\label{eq:pop-max-null}
\max\{m_1(\theta_0),m_2(\theta_0),0\}=0.
\end{equation}

Recall the sample moments used in $T_n(\theta)$:
\[
\hat m_{n,1}(\theta)
=
\widehat{\Delta E}_n(\theta)-\delta_h(j)\sqrt{\widehat V_b},\qquad
\hat m_{n,2}(\theta)
=
-\widehat{\Delta E}_n(\theta)-\delta_h(j)\sqrt{\widehat V_b}.
\]
By Assumption~\ref{ass:inference}.\ref{ass:inf-rate} and $\widehat V_b\xrightarrow{p}V_b$,
\begin{equation}\label{eq:moment-expansion}
\sqrt{n}\Big(\hat m_{n,\ell}(\theta_0)-m_\ell(\theta_0)\Big)
=
\sqrt{n}\Big(\widehat{\Delta E}_n(\theta_0)-\Delta E[\psi_b(u)\tilde W_j(\theta_0)]\Big)
+o_p(1),
\qquad \ell=1,2.
\end{equation}
Moreover, Assumptions~\ref{ass:inference}.\ref{ass:inf-cf}--\ref{ass:inf-rate}
imply that the leading term on the right admits an asymptotically linear representation:
there exist i.i.d.\ mean-zero variables $\{\varphi_i(\theta_0)\}$ with finite variance such that
\begin{equation}\label{eq:asylin}
\sqrt{n}\Big(\widehat{\Delta E}_n(\theta_0)-\Delta E[\psi_b(u)\tilde W_j(\theta_0)]\Big)
=
\frac{1}{\sqrt{n}}\sum_{i=1}^n \varphi_i(\theta_0)+o_p(1).
\end{equation}
Combining \eqref{eq:moment-expansion} and \eqref{eq:asylin} yields, for $\ell=1,2$,
\begin{equation}\label{eq:moment-asylin}
\sqrt{n}\Big(\hat m_{n,\ell}(\theta_0)-m_\ell(\theta_0)\Big)
=
\frac{1}{\sqrt{n}}\sum_{i=1}^n \varphi_i(\theta_0)+o_p(1).
\end{equation}

By definition,
\[
T_n(\theta_0)
=
\sqrt{n}\max\{\hat m_{n,1}(\theta_0),\hat m_{n,2}(\theta_0),0\}.
\]
Using $m_\ell(\theta_0)\le 0$ and \eqref{eq:moment-asylin},
\[
T_n(\theta_0)
=
\max\Big\{
\sqrt{n}\big(\hat m_{n,1}(\theta_0)-m_1(\theta_0)\big)+\sqrt{n}m_1(\theta_0),\ 
\sqrt{n}\big(\hat m_{n,2}(\theta_0)-m_2(\theta_0)\big)+\sqrt{n}m_2(\theta_0),\
0
\Big\}.
\]
Since $\sqrt{n}m_\ell(\theta_0)\le 0$, this implies the stochastic domination
\begin{equation}\label{eq:domination}
T_n(\theta_0)
\le
\max\Big\{
\sqrt{n}\big(\hat m_{n,1}(\theta_0)-m_1(\theta_0)\big),\
\sqrt{n}\big(\hat m_{n,2}(\theta_0)-m_2(\theta_0)\big),\
0
\Big\}.
\end{equation}
By \eqref{eq:moment-asylin}, the right-hand side equals
\[
\max\Big\{
\frac{1}{\sqrt{n}}\sum_{i=1}^n \varphi_i(\theta_0),\
-\frac{1}{\sqrt{n}}\sum_{i=1}^n \varphi_i(\theta_0),\
0
\Big\}
+o_p(1),
\]
which is the max of a centered asymptotically normal statistic and its negative.
Assumption~\ref{ass:inference}.\ref{ass:inf-nondeg} ensures non-degeneracy.

Let $\{\xi_i\}_{i=1}^n$ be i.i.d.\ $N(0,1)$ multipliers independent of the data,
and let $\hat\varphi_i(\theta_0)$ be the cross-fitted estimate of $\varphi_i(\theta_0)$.
Define
\[
\widehat{\Delta E}_n^*(\theta_0)
:=
\frac{1}{\sqrt{n}}\sum_{i=1}^n \xi_i \hat\varphi_i(\theta_0),
\qquad
T_n^*(\theta_0)
:=
\max\{\widehat{\Delta E}_n^*(\theta_0),-\widehat{\Delta E}_n^*(\theta_0),0\}.
\]
By Assumption~\ref{ass:inference}.\ref{ass:inf-bootstrap},
the conditional distribution of $T_n^*(\theta_0)$ given the data consistently
approximates the distribution of the dominating statistic in
\eqref{eq:domination}. Hence, if $c_{1-\alpha_G}(\theta_0)$ denotes the
conditional $(1-\alpha_G)$ quantile of $T_n^*(\theta_0)$, we have
\begin{equation}\label{eq:size-control-fixed}
\limsup_{n\to\infty} P\big(T_n(\theta_0)>c_{1-\alpha_G}(\theta_0)\big)
\le \alpha_G.
\end{equation}

By definition,
$\theta_0\notin\mathcal C_{1-\alpha_G}(j,b,e_0)$ iff
$T_n(\theta_0)>c_{1-\alpha_G}(\theta_0)$.
Therefore \eqref{eq:size-control-fixed} implies
\begin{equation}\label{eq:coverage-fixed}
\liminf_{n\to\infty}P\big(\theta_0\in\mathcal C_{1-\alpha_G}(j,b,e_0)\big)
\ge 1-\alpha_G.
\end{equation}

Since $\mathcal G$ is finite, by the union bound,
\begin{align*}
P\big(\theta_0\notin \mathcal C_{1-\alpha}\big)
&=
P\Big(\theta_0\notin \bigcap_{(j,b,e_0)\in\mathcal G}\mathcal C_{1-\alpha_G}(j,b,e_0)\Big)
=
P\Big(\bigcup_{(j,b,e_0)\in\mathcal G}\{\theta_0\notin \mathcal C_{1-\alpha_G}(j,b,e_0)\}\Big)\\
&\le
\sum_{(j,b,e_0)\in\mathcal G}
P\big(\theta_0\notin \mathcal C_{1-\alpha_G}(j,b,e_0)\big).
\end{align*}
Taking $\limsup_{n\to\infty}$ on both sides and using \eqref{eq:coverage-fixed},
\[
\limsup_{n\to\infty}P\big(\theta_0\notin \mathcal C_{1-\alpha}\big)
\le
\sum_{(j,b,e_0)\in\mathcal G} \alpha_G
=
|\mathcal G|\cdot \frac{\alpha}{|\mathcal G|}
=
\alpha,
\]
which is equivalent to
$\liminf_{n\to\infty}P(\theta_0\in\mathcal C_{1-\alpha})\ge 1-\alpha$.
This proves the theorem.
\end{proof}

\begin{remark}
    In applications, $(j,b)$ govern a bias--variance tradeoff: larger $j$ reduces
sieve bias $\delta_h(j)$, while smaller $b$ localizes the density-ratio
approximation but reduces effective sample size. The localization point $e_0$
can be set on a grid spanning the support of $\hat\varepsilon_2$.
\end{remark}
\begin{remark}
The procedure profiles over nuisance parameters used to construct
$\hat h_j$ and $\hat\ell$. In particular, $\hat h_j$ depends on the candidate
$\gamma_1$ through \eqref{eq:alpha-hat}, and $\hat\ell$ depends on $(b,e_0)$.
\end{remark}


\begin{proposition}[Identification of $\beta_1$ given $\gamma_1$]
\label{prop:beta1-id}
Suppose Assumption~\ref{ass:baseline} holds and $\gamma_1$ is known.
Then $\beta_1$ is point identified and satisfies
\begin{equation}
\beta_1
=
\Sigma_{XX}^{-1}
E\!\left[
X\big(Y_1-\gamma_1 Y_2\big)
\right],
\qquad
\Sigma_{XX}:=E[XX'].
\label{eq:beta1-id}
\end{equation}
\end{proposition}

\begin{proof}
Substituting \eqref{eq:control-decomp} into \eqref{eq:secondstage},
\[
Y_1-\gamma_1 Y_2
=
X'\beta_1
+
h(\varepsilon_2)
+
\eta.
\]
Multiplying both sides by $X$ and taking expectations yields
\[
E\!\left[
X(Y_1-\gamma_1 Y_2)
\right]
=
E[XX']\beta_1
+
E[X h(\varepsilon_2)]
+
E[X\eta].
\]
Under Assumption~\ref{ass:baseline}.\ref{ass:exoX},
$E[\varepsilon_2\mid X]=0$ and $E[\eta\mid X]=0$.
Since $h(\varepsilon_2)$ is measurable in $\varepsilon_2$,
it follows that $E[X h(\varepsilon_2)]=0$ and $E[X\eta]=0$.
Therefore,
\[
E[X(Y_1-\gamma_1 Y_2)]
=
\Sigma_{XX}\beta_1.
\]
By Assumption~\ref{ass:baseline}.\ref{ass:pd},
$\Sigma_{XX}$ is positive definite and invertible,
which yields \eqref{eq:beta1-id}.
\end{proof}

Proposition~\ref{prop:beta1-id} shows that it is reasonable that the identification problem is entirely concentrated on $\gamma_1$.
Once $\gamma_1$ is determined (point or set identified), $\beta_1$ follows from a standard linear projection argument.

\subsection{Orthogonality and robustness to nuisance estimation}

The construction of $\widehat{\Delta E}_n(\theta)$ involves several
generated objects: the first-stage residual $\hat\varepsilon_2$,
the sieve estimator $\hat h_j$, and the local density-ratio estimator
$\hat\ell$. We now clarify why these nuisance estimators do not affect
the first-order asymptotic behavior of the test statistic.

At the true parameter $\theta_0$, the population moment satisfies
\[
\Delta E[\psi_b(u)\tilde W_j(\theta_0)]
=
E_0[\psi_b(u) r_j(u)(e^{\ell(u)}-1)],
\]
since $E[\eta \mid X,\varepsilon_2]=0$ implies
$E_0[\psi_b(u)\eta(e^{\ell(u)}-1)]=0$.
Hence, the contribution of the structural disturbance $\eta$
is orthogonal to any measurable function of $u$.

This implies a local orthogonality property:
small perturbations of the nuisance components
$(\beta_2,h_j,\ell)$ affect the moment only through the
approximation error $r_j(u)$ and higher-order terms.
Formally, the Gateaux derivative of the moment with respect to
$n^{-1/2}$-scale perturbations of the nuisance functions
vanishes at $\theta_0$.

Under Assumption~\ref{ass:baseline}.\ref{ass:exoX},
$\hat\beta_2-\beta_2=O_p(n^{-1/2})$,
and therefore
\[
\hat u_i-u_i
=
X_i'(\beta_2-\hat\beta_2)
=
O_p(n^{-1/2}).
\]
Since $\psi_b(\cdot)$ and the sieve basis functions are bounded
and Lipschitz on $\mathcal U_b$,
the substitution of $\hat u_i$ for $u_i$
induces an $O_p(n^{-1/2})$ perturbation in
$\widehat{\Delta E}_n(\theta)$.
Consequently, the generated-regressor effect is asymptotically negligible
relative to the $n^{-1/2}$ stochastic fluctuation of the empirical process.

The estimator $\hat\ell$ enters the procedure only through
$\widehat V_b$, which affects the inequality bound
$\delta_h(j)\sqrt{\widehat V_b}$.
Under consistency of $\hat\ell$ and $\widehat V_b \xrightarrow{p} V_b$,
the plug-in error in $\sqrt{\widehat V_b}$
is $o_p(1)$.
Since the test statistic scales the sample moment by $\sqrt{n}$,
this plug-in effect is second-order and does not alter
the limiting distribution under the null.

To mitigate overfitting bias in $\hat h_j$ and $\hat\ell$,
we employ cross-fitting.
Conditional on the training folds,
the evaluation-fold moments are computed using nuisance estimates
constructed from independent subsamples.
This ensures that first-order empirical process fluctuations
remain centered and that nuisance estimation errors contribute
only through higher-order terms.

Together, these arguments justify treating
$\widehat{\Delta E}_n(\theta)$ as asymptotically linear
with influence function depending only on the data,
as stated in Assumption~\ref{ass:inference}.

\section{Conclusion}

This paper studies identification and inference in a triangular system with an endogenous regressor when the researcher does not impose a classical exclusion restriction and allows for a flexible, nonparametric control function. In this setting, baseline orthogonality conditions are generally insufficient for point identification, as the unknown control function can absorb the dependence that generates endogeneity.

We show that informative restrictions can nevertheless be obtained by exploiting distributional shifts in the first-stage disturbance induced by an auxiliary variable. Under a local restriction on the shape of the log density ratio and an explicit bound on the sieve approximation error of the control function, we derive moment inequalities that restrict the structural parameter of interest. The strength of identification depends transparently on two elements: the local relevance of the distributional shift and the quality of the sieve approximation. When the approximation error vanishes and local relevance holds, the identified set collapses to a point.

We also develop a fully operational inference procedure based on test inversion and multiplier bootstrap. The method accommodates generated regressors, cross-fitted sieve estimation, and locally estimated density-ratio nuisances while maintaining valid asymptotic coverage. The framework makes explicit how identification strength varies with tuning parameters and sensitivity bounds.

More broadly, the results illustrate that even in the absence of exclusion restrictions, partial and sometimes point identification can be recovered from weak distributional structure. This perspective complements existing approaches based on heteroskedasticity or parametric variance restrictions, and highlights the role of local distributional assumptions as a source of identifying power. Future work may extend the framework to multivalued or continuous auxiliary variables, nonlinear outcome equations, or high-dimensional covariates.

\bibliographystyle{ecta}
\bibliography{citation.bib}

@article{lewbel2012using,
  title={Using heteroscedasticity to identify and estimate mismeasured and endogenous regressor models},
  author={Lewbel, Arthur},
  journal={Journal of business \& economic statistics},
  volume={30},
  number={1},
  pages={67--80},
  year={2012},
  publisher={Taylor \& Francis}
}

@article{klein2010estimating,
  title={Estimating a class of triangular simultaneous equations models without exclusion restrictions},
  author={Klein, Roger and Vella, Francis},
  journal={Journal of Econometrics},
  volume={154},
  number={2},
  pages={154--164},
  year={2010},
  publisher={Elsevier}
}

\end{document}